%% file: forarxiv_camera_ready.tex
\documentclass{article}
\usepackage{ijcai26}

\usepackage{times}
\usepackage{soul}
\usepackage{url}
\usepackage[hidelinks]{hyperref}
\usepackage[utf8]{inputenc}
\usepackage[small]{caption}
\usepackage{graphicx}
\usepackage{amsmath}
\usepackage{amsthm}
\usepackage{booktabs}
\usepackage[switch]{lineno}

\input{macros}

\title{Stabilizing Welfare-Maximizing Decisions via Endogenous Transfers}

\author{
    Joshua Kavner
    \affiliations
    Independent Researcher
    \emails
    joshuakavner.research@gmail.com
}

\begin{document}

\maketitle

\begin{abstract}
Many multi-agent systems depend on collective decisions made by self-interested agents, which raises deep questions about coalition formation and stability. We study social choice with endogenous, outcome-contingent transfers, where agents voluntarily form contracts that redistribute utility depending on the collective decision, allowing for fully strategic coalition formation. We show that under consensus rules, individually rational strong Nash equilibria (IR-SNE) always exist, implementing welfare-maximizing outcomes with feasible transfers, and provide a simple, efficient algorithm to construct them. For more general anonymous, monotonic, and resolute rules, we identify necessary conditions for viable deviations, significantly limiting the possibility of destabilizing coalitions. By bridging cooperative and noncooperative perspectives, our approach shows that transferable utility can achieve core-like stability, restoring efficiency and budget balance even where classical impossibility results generally apply. Overall, this framework offers a practical and robust way to coordinate large-scale strategic multi-agent systems.
\end{abstract}

\section{Introduction}

It is 2050, and the utopia of the Internet of Things has become a reality in transportation. Autonomous vehicles taxi commuters to work, transcontinental shipping is coordinated through autonomous platoons, and algorithms balance intrinsic tradeoffs between safety and efficiency. From a policy maker’s perspective, centralized solutions are appealing: by collecting global information, algorithms can guide traffic flows toward socially optimal outcomes, reduce congestion, and avoid the inefficiencies that arise when agents act myopically (e.g., Braess's paradox \citep{christodoulou2005price}). However, realizing this vision in practice is far from trivial. Computing socially optimal outcomes is often intractable, incentive-compatible preference elicitation is possible only under restrictive conditions, and centralized mechanisms are fragile to communication failures. As a result, purely centralized solutions struggle to deliver robust, reliable performance.

These limitations have motivated research on decentralized multi-agent systems, where self-interested agents act autonomously but coordinate through learning or negotiation \citep{kraus1997negotiation,ephrati1997heuristic}. For example, centralized training with decentralized execution (CTDE) allows agents to learn coordinated policies under shared information during training \citep{lowe2017multi,amato2024introduction}. However, once deployed, agents cannot adapt or negotiate on the fly. In many strategic environments, agents must alternate between competition and cooperation depending on context \citep{axelrod1981evolution}. Coalition formation thus provides a natural middle ground: by forming voluntary coalitions, agents can improve outcomes relative to fully decentralized behavior without requiring full centralization. Still, coalition formation comes with challenges: computing optimal coalitions is NP-hard \citep{sandholm1999coalition}, iterative coalition deviations may fail to converge \citep{meir2017iterative}, and strategic behavior under common social choice rules remains computationally intractable \citep{xia2009complexity}.

In this paper, we propose \emph{endogenous transfers} as a mechanism for stabilizing collective decision-making in multi-agent systems. Collective choices among self-interested agents can naturally be modeled as a social choice problem: agents vote over alternatives, and the outcome depends on the aggregation of their preferences. Rather than prescribing coalitions or enforcing outcomes, our approach allows agents to voluntarily redistribute utility through contracts that are contingent on the chosen alternative. These transfers reshape incentives and expand the space of stable outcomes, allowing coalitions to form and dissolve strategically \citep{jackson2005endogenous}.

Our model is inspired by classical work on bribery and coalitional manipulation in voting, as well as cooperative game-theoretic analyses of social choice \citep{bachrach2011coalitional}. In traditional bribery or manipulation models, coalitions are often assumed to be exogenous and bribes are budget-constrained \citep{elkind2009swap,faliszewski2009hard}. By contrast, in our framework, both coalitions and transfers are fully endogenous: agents decide how to vote and how to contract, anticipating the behavior of others while respecting individual rationality. This implies that agents can offer transfers contingent on a particular alternative winning, forming stable coalitions without any central enforcement.

The central question, then, is whether such voluntary transfers are sufficient to stabilize welfare-maximizing outcomes in the presence of strategic, coalitional behavior. We show that under consensus rules, stable outcomes always exist and can be constructed efficiently, and we characterize fundamental limits on stability under more general voting rules.

\subsection{Our Contributions}

We make the following contributions.

\paragraph{A model of endogenous transfers in social choice.}
We introduce a noncooperative model in which agents may form transfer agreements
in a strategic pre-round, followed by strategic voting under a fixed social
choice rule.
Contracts are conditioned on alternatives winning and allow agents to share the
utility generated by collective decisions.
This framework separates contract formation from preference revelation, enabling
the study of coalition stability without assuming centralized enforcement or
exogenously given coalitions.

\paragraph{Strong Nash equilibrium under consensus rules.}
For consensus rules, which generalize cooperative bargaining scenarios in which
agents can veto unfavorable outcomes, we show that individually rational strong
Nash equilibria (IR-SNE) always exist.
Moreover, every IR-SNE induces a welfare-maximizing outcome supported by a
feasible transfer scheme.
We provide a simple, computationally efficient algorithm for constructing such
equilibria, establishing both existence and implementability.

\paragraph{Necessary conditions for deviations in general.}
We extend our analysis beyond consensus to arbitrary anonymous, monotonic, and
resolute social choice rules.
In this setting, we derive a tight necessary condition for the existence of
IR coalitional deviations.
While this does not resolve known NP-hardness of manipulation for fixed
rules, it substantially restricts the space of profitable deviations and clarifies
when stability cannot be restored via transfers.

\paragraph{A bridge between cooperative and noncooperative social choice.}
Our results admit a natural interpretation in terms of cooperative game theory.
Although the underlying environment is noncooperative, the presence of
transferable utility allows coalitions to redistribute surplus via contracts.
As a result, strong Nash equilibrium coincides with a notion of core stability
with respect to the set of outcomes implementable by the social choice rule.
In particular, under consensus, IR-SNE implements core-stable,
welfare-maximizing allocations despite the absence of centralized mechanism
design.

Finally, our framework offers a complementary perspective on classical
impossibility results in mechanism design.
\citeay{green1977characterization} show that no mechanism can simultaneously achieve
efficiency, budget balance, and strategyproofness in general quasi-linear
domains \citep{nath2019efficiency}.
Rather than contradicting this result, we refine the notion of strategic
robustness.
By allowing transferable utility and coalition-aware contracts, we replace
unconditional truthfulness with a coalition-proof requirement: no group of
agents can profitably deviate given feasible transfers among its members.
In this sense, IR-SNE can be viewed as a refinement of strategyproofness that
restores efficiency and budget balance within the space of coalitionally stable
outcomes.

\subsection{Related Work}

The computational complexity of bribery in elections was first studied by \citeay{faliszewski2006complexity}, who asked whether a preferred candidate can be made the winner by modifying the preferences of a limited number of agents. Subsequent work analyzed bribery under many voting rules, including Borda \citep{brelsford2008approximability}, maximin \citep{faliszewski2011multimode}, Schulze \citep{parkes2012complexity}, $k$-approval \citep{lin2012solving}, STV and ranked pairs \citep{xia2012computing}, and Bucklin and fallback voting \citep{faliszewski2015complexity}. \citeay{faliszewski2008nonuniform} introduced non-uniform bribery prices per agent under a fixed budget, while \citeay{Faliszewski2007LlullAC} and \citeay{elkind2009swap} studied models where the cost of a bribe depends on the swap distance between the original and modified rankings. For comprehensive overviews, see \citeay{baumeister2016preference} and \citeay{faliszewski2016control}.

Bribery is closely related to coalitional manipulation, which studies whether a preferred outcome can be achieved by coordinated misreports from a coalition of agents \citep{conitzer2007elections,xia2009complexity}. Unlike bribery, coalitional manipulation typically assumes the coalition is exogenously given; endogenously selecting coalitions and distributing gains is itself computationally intractable \citep{sandholm1999coalition,chalkiadakis2011computational}. \citeay{bachrach2009cost} introduced the cost of stability, and \citeay{bachrach2011coalitional} showed that bribery can be interpreted as utility-constrained coalition recruitment, linking voting with cooperative game theory.
Our model extends this perspective by explicitly separating preference reports from transfer contracts and by imposing individual rationality on every participating agent, leading to the notion of IR-SNE. This allows us to construct feasible, budget-balanced transfer schemes that stabilize outcomes under consensus and related rules, complementing Bachrach et al.’s existential and complexity results with a constructive, mechanism-aware analysis.

Related questions of side payments and commitments have been studied in game theory \citep{fershtman1991observable,jackson2005endogenous,monderer2009strong,kalai2013cooperation,turrini2013endogenous,grandi2019negotiable,geffner2025maximizing} and multi-agent reinforcement learning \citep{sodomka2013coco,yang2020learning,haupt2024formal,kolumbus2024paying}. A separate literature examines vote trading through competitive equilibrium models \citep{casella2012competitive,casella2021trading,casella2021does}.

\section{Preliminaries}

\begin{paragraph}{Model.}
An instance 
$I = \langle \mathcal{N}, \mathcal{A}, u \rangle$
consists of a set of $n$ agents $\mathcal{N}$, a set of $m$ alternatives 
$\mathcal{A}$, and base utility functions 
$u = (u_1,\ldots,u_n)$, where each
$u_i : \mathcal{A} \to \mathbb{R}_{\ge 0}$. We denote the \emph{social welfare} of an alternative $a$ as $SW(a) = \sum_{i \in \calN} u_i(a)$.
A resolute social choice function $f : \mathcal{L}(\mathcal{A})^n \to \mathcal{A}$
selects a single alternative given votes of the agents. We assume $f$ is anonymous, monotone, and resolute (AMR); it uses a fixed tie-breaking order.

In particular, we focus on the resolute \emph{consensus rule}, where an alternative is chosen only if it is top-ranked by all agents; otherwise, a default alternative $a^* \in \calA$ is selected. Consensus is both anonymous and monotonic.
\end{paragraph}

\begin{paragraph}{Contracts.}
A \emph{contract} is a collection of transfers 
$c:\mathcal N\times\mathcal N\times\mathcal A\to\mathbb R$, where
$c_{i\to j}(a)$ represents a promise from agent $i$ to agent $j$ conditional on alternative $a$.
Contracts are proposed before voting and are enforced conditional on the realized winner.
Each contract induces a net transfer scheme $\tau$ defined by
$\tau_i(a) = \sum_{j\neq i} c_{j\to i}(a) - \sum_{j\neq i} c_{i\to j}(a)$, which satisfies budget balance: $\sum_i \tau_i(a)=0$. 
We assume that the set of contracts is closed under netting, meaning that only the induced transfer scheme $\tau$
affects utilities and strategic behavior, except when the identity of transfer recipients is relevant for coalition formation. The null contract is denoted by $\emptyset$.
\end{paragraph}

\begin{paragraph}{Strategic Voting.}
Given base utilities $u$ and transfers $\tau$, agents' realized
utilities are quasi-linear: if alternative $a$ is chosen, agent $i$
receives $U_i(a)=u_i(a)+\tau_i(a)$. These utilities induce a truthful preference ordering $R_i^*(u+\tau)\in\mathcal L(\mathcal A)$ defined by $a \succ_i^* b \quad \Longleftrightarrow \quad U_i(a) > U_i(b)$,
with ties broken according to a fixed ordering over $\mathcal A$.
A vote $R_i$ is \emph{truthful} if $R_i = R_i^*(u+\tau)$.
Agents submit votes $R_i \in \mathcal L(\mathcal A)$, yielding a vote
profile $P=(R_1,\ldots,R_n)$. For a subset of agents
$S \subseteq \mathcal N$, we denote by
$P_S^*(u+\tau)$ the truthful vote profile of agents in $S$ under
utilities $u+\tau$. We write
$P^* = P_{\mathcal N}^*(u)$ for the fully truthful profile under the
null contract.
\end{paragraph}

\begin{paragraph}{Coalitional Deviation.}
Our definition of coalitional deviation builds on concepts from \emph{unweighted coalitional manipulation} \citep{xia2009complexity} and \emph{bribery} \citep{faliszewski2009hard}, both of which assume truthful behavior by non-deviating agents. These definitions require refinement when stability is evaluated over states $(P,\tau,S)$ or along sequences of deviations, as in iterative voting \citep{meir2017iterative}. In particular, only agents in the deviating coalition are allowed to modify outgoing transfers or vote strategically. Agents outside the coalition may experience changes in their incoming transfers due to the coalition's modifications but do not actively adjust their own outgoing contracts.

\begin{dfn}[Coalitional Deviation]
Given base utilities $u$, a \emph{coalitional deviation} from a state $(P,\tau,S)$ is any state $(P',\tau',S')$ constructed as follows. A coalition of agents $S' \subseteq \calN$ proposes modifications to outgoing contracts $(c'_{i\rightarrow j}(a))_{i \in S', j \in \mathcal{N}, a \in \mathcal{A}}$ relative to the existing contract scheme $c$, inducing a new transfer scheme $\tau'$, and submits new votes $P_{S'}$. Agents outside the coalition, $i \notin S'$, are assumed to vote truthfully with respect to their updated utilities $u+\tau'$, denoted $P^*_{-S'}(u+\tau')$. The resulting vote profile is $P' = (P_{S'}, P^*_{-S'}(u+\tau'))$, yielding outcome $f(P')$.
\label{dfn:coalitional_manipulation}
\end{dfn}

Under Definition~\ref{dfn:coalitional_manipulation}, an agent $i$ belongs to the deviating coalition $S'$ if and only if they either vote strategically with respect to $u+\tau'$, or actively modify their outgoing contracts. This notion is \emph{maximal}: agents may join a coalition in favor of an alternative even if they are not individually pivotal. Non-participating agents vote truthfully with respect to the updated utilities $u+\tau'$, modeling the independence of coalitions: each coalition chooses its deviation assuming all other agents respond truthfully to their adjusted utilities. 
This contrasts with an alternative, \emph{sticky} notion, where non-participating agents neither update their votes nor their contracts. In Appendix~A, we show that equilibrium may fail to exist under this sticky variant, even in trivial settings.
\end{paragraph}

\begin{dfn}[Individually Rational Deviation]
Let $(P,\tau,S)$ be a state and $(P',\tau',S')$ a coalitional deviation. The deviation $(P',\tau',S')$ is \emph{individually rational (IR)} if
\[
U'_i(f(P')) \ge U_i(f(P)) \quad \forall i \in S',
\]
with strict inequality for at least one $i \in S'$. We call $(P', \tau', S')$ \emph{IR-feasible} from $(P, \tau, S)$.
\end{dfn}

\begin{dfn}[Strong Nash Equilibrium]
A state $(P,\tau,S)$ is a \emph{strong Nash equilibrium (SNE)} if no coalition $S'$ has an IR deviation. It is an \emph{IR-SNE} if it is additionally IR-feasible from the truthful state $(P^*,\emptyset,\emptyset)$.
\end{dfn}

\begin{examplex}
Consider alternatives $\calA = \{1,2\}$, agents $\calN = \{1, \ldots, 5\}$, $f$ as the consensus rule with default $a^* = 2$, and utilities $u = (u_i(a))_{i \in \calN, a \in \calA}$ defined as:
\[
u = \begin{pmatrix}
2 & 1 & 2 & 1 & 11 \\
4 & 1 & 3 & 2 & 3
\end{pmatrix}.
\]
The truthful vote is $R^*_5 = (1 \succ 2)$ and $R^*_i = (2 \succ 1)$, $i \in \calN \setminus \{5\}$, so that $f$ returns $2$ by default. A feasible transfer scheme is
\[
\tau = \begin{pmatrix}
2 & 0 & 1 & 1 & -4 \\
0 & 0 & 0 & 0 & 0
\end{pmatrix} \implies u+ \tau = \begin{pmatrix}
4 & 1 & 3 & 2 & 7 \\
4 & 1 & 3 & 2 & 3
\end{pmatrix}.
\]
In $u+\tau$, all agents $i \in \calN$ vote with $R_i = (1 \succ 2)$, with $P = (R_i)_{i \in \calN}$ so that $f(P) = 1$. The set of participating agents $S = \calN$. All agents in $S$ weakly prefer alternative $1$ in $u+\tau$ to $2$ in $u$, so the state $(P, \tau, S)$ is IR-feasible from the truthful state $(P^*, \emptyset, \emptyset)$. It can be shown that there is no $(P', \tau', S')$ so that agents in $S'$ weakly prefer $u_i(f(P'))+\tau'_i(f(P'))$ to $u_i(f(P))+\tau_i(f(P))$, so $(P, \tau, S)$ is an IR-SNE.
\end{examplex}

\section{Main Results}
We focus on the consensus rule $f$, where an alternative $a \in \calA$ is chosen if and only if all agents rank $a$ first in their reported preferences; otherwise, a default $a^* \in \calA$ is selected. This setting is motivated both practically and technically. Practically, consensus generalizes cooperative bargaining scenarios, in which multiple agents must determine how to split a surplus across a Pareto frontier: what benefits one agent necessarily limits another. In these settings, the default $a^*$ represents a disagreement value that agents receive if negotiations fail \citep{thomson1994cooperative}. Consensus thus provides a concise description of real-world multi-agent planning problems in which any agent can veto a plan, forcing all agents to receive $a^*$ \citep{kraus1997negotiation}.

Technically, consensus is appealing because it allows for concise necessary and sufficient conditions for SNE, and its structure provides intuition that extends to other anonymous, monotonic, and resolute (AMR) rules. To see why, consider a state $(P, \tau, S)$ in a transferable-utility social choice game. A potential SNE can fail in exactly two ways: any coalitional deviation $(P', \tau', S')$ starting from $(P, \tau, S)$ with $f(P) = b$ either (i) changes the outcome, with $f(P') \neq b$, or (ii) preserves the outcome, with $f(P') = b$ but redistributes utilities in a way that strictly benefits the coalition.

Under consensus, the first type of deviation is tightly constrained: by Lemma~\ref{lem:no_Astar_transfer}, there is no IR-feasible state in which the default alternative \(a^*\) becomes the winner, a consequence of the \emph{full coverage} requirement (Definition~\ref{dfn:full_coverage}) that all agents rank a non-default alternative first. The second type of deviation preserves the outcome but reallocates transfers; Lemma~\ref{lem:sw_max_swap} characterizes when such redistributions are compatible with IR. Together, these lemmas pin down the structure of IR-SNE under consensus and illustrate general principles for AMR rules: welfare maximization is necessary for equilibrium (Proposition~\ref{prop:non_welfare_max_is_not_SNE}); outcome-changing deviations require sufficiently large reallocations of utility (Proposition~\ref{prop:stability_via_indifferent_donors}); and outcome-preserving deviations impose stricter constraints on feasible transfers than those characterized in Lemma~\ref{lem:sw_max_swap}.

We now turn to the technical groundwork that allows us to formalize these ideas, starting with the role of truthful states and the observability of agent participation.

\subsection{Truthful Considerations and Observable Participation}
\label{sec:truthful_considerations}
Many of our proofs focus on the contract-adjusted utilities $\tau$, rather than the underlying individual transfers $(c_{i \rightarrow j}(a))_{i,j \in \calN, a \in \calA}$. We assume throughout that net transfers along transitive chains cancel out, so that an agent’s utility depends only on the net transfer $\tau_i(a)$ when evaluating overall stability. This allows us to reason in terms of $\tau$ rather than tracking the full contract matrix $c$.

We emphasize, however, that this simplification is a sleight of hand: if the underlying contracts $c$ are ignored entirely, so that agents only care how much they donate without regard to whom, this can break equilibrium, as shown in Appendix~A. Lemma~\ref{lem:observable_participation} is a direct corollary of this simplification: it lets us deduce whether $i \in S$ from changes in $\tau$, but this reasoning is valid only because $\tau$ captures the net effect of transitive contracts $c$.

\begin{lem}[Observable Participation]
Given base utilities $u$, consider states $(P, \tau, S)$ and $(P', \tau', S')$. If there exists an alternative
$a \in \mathcal{A}$ such that either (1) $\tau_i(a) \le 0$ and $\tau'_i(a) < \tau_i(a)$, or (2) $\tau_i(a) \ge 0$ and $\tau'_i(a) < 0$, then $i \in S'$.
\label{lem:observable_participation}
\end{lem}

\begin{proof}
In both cases, agent $i$ experiences a strictly negative change in payoff under $a \in \calA$. Since non-participating agents of $S'$ cannot actively modify contracts, and net transfers cancel along chains by assumption, such changes can only result from $i$’s own strategic actions, implying $i \in S'$.
\end{proof}


We can therefore identify coalition members based on changes in net transfers. This reasoning allows us to analyze potential deviations in general AMR rules: if the chosen alternative $f(P)$ is not a social welfare maximizer, some coalition of agents can strictly increase their total utility by jointly reallocating transfers toward a welfare-maximizing alternative. Thus, non-welfare-maximizing states cannot be SNE. 

\begin{prop}[Non-Welfare-Maximizing States Are Not SNE]
\label{prop:non_welfare_max_is_not_SNE}
Given base utilities $u$ and an AMR rule $f$, consider a state $(P, \tau, S)$. If $f(P) \notin \arg\max_{a \in \calA} SW(a)$, then $(P, \tau, S)$ is not an SNE.
\end{prop}

\begin{proofsketch}
If $f(P)$ is not a social welfare maximizer, the grand coalition can jointly reallocate transfers toward an alternative $b \in \arg\max_{a \in \calA} SW(a)$ such that $b$ fully covers $f(P)$. This deviation weakly increases total utility for the chosen alternative for all agents, showing that $(P, \tau, S)$ cannot be stable. The full proof is provided in Appendix~B.2.
\end{proofsketch}

Consensus strengthens the connection between social welfare maximization and equilibrium stability through the notion of \emph{full coverage}. While Proposition~\ref{prop:non_welfare_max_is_not_SNE} shows that non-welfare-maximizing outcomes cannot be stable under any AMR rule, consensus imposes a stricter condition: for a non-default alternative $b$ to win, every agent must rank $b$ first, meaning $b$ must fully cover the truthful winner $f(P^*)$ in post-transfer utilities. This ensures that no coalition can profitably deviate, making social welfare maximization a necessary condition for IR-SNE under consensus.

\begin{dfn}[Fully Covered]
Given utilities $u$ and a transfer scheme $\tau$, an alternative $b \in \mathcal{A}$ is said to be \emph{fully covered} by $u+\tau$ if
\[
U_i(b) \ge U_i(a), \quad \forall i \in \mathcal{N}, \forall a \in \mathcal{A}.
\]
\label{dfn:full_coverage}
\end{dfn}

\begin{prop}[Truthful State Is SNE Under Consensus]
Given base utilities $u$ and a resolute consensus rule $f$, consider the truthful state $(P^*, \emptyset, \emptyset)$. If $f(P^*) \in \arg\max_{a \in \mathcal{A}} SW(a)$, then $(P^*, \emptyset, \emptyset)$ is an SNE.
\label{prop:truthful_is_SNE}
\end{prop}

\begin{proofsketch}
By IR and the definition of consensus, any alternative in a potential deviation must fully cover the truthful winner $f(P^*)$. Since $f(P^*)$ is a social welfare maximizer, only other welfare-maximizing alternatives can meet this requirement. Any redistribution that changes participating agents’ utilities would violate IR, so no IR-feasible deviation exists. Hence, $(P^*, \emptyset, \emptyset)$ is an SNE. The full proof is provided in Appendix~B.2.
\end{proofsketch}

\subsection{IR-SNE for Consensus}
\label{sec:IR_SNE_consensus}

Building on the discussion of truthful states and full coverage in Sec.~\ref{sec:truthful_considerations}, we turn to the construction and characterization of IR-SNE under the consensus rule \(f\). In such equilibria, the selected alternative \(f(P)\) must maximize social welfare and induce full coverage in post-transfer utilities \(u+\tau\), which tightly restricts IR-feasible deviations.

Within this framework, Lemma~\ref{lem:no_Astar_transfer} establishes that no IR-feasible deviation can make the default alternative \(a^*\) a winner, while Lemma~\ref{lem:sw_max_swap} characterizes when utilities associated with a welfare-maximizing alternative can be redistributed without violating IR. Together, these results pin down the constraints governing IR-SNE under consensus.

\begin{lem}
Given base utilities $u$ and consensus rule $f$, consider a state $(P,\tau,S)$ that is IR-feasible from the truthful state $(P^*,\emptyset,\emptyset)$.
Suppose that $f(P^*) = a^* \notin \arg\max_{a \in \calA} SW(a)$, that $f(P)=b \in \arg\max_{a \in \calA} SW(a)$, and that $b$ has full coverage under $u+\tau$. Then there is no state $(P', \tau', S')$ with $f(P') = a^*$ that is IR-feasible from $(P, \tau, S)$.
\label{lem:no_Astar_transfer}
\end{lem}

\begin{proofsketch}
The argument proceeds by contradiction. Suppose that starting from a state $(P,\tau,S)$ in which the welfare-maximizing alternative $b$ is implemented with full coverage, there exists an IR-feasible deviation to a state $(P',\tau',S')$ that implements $a^*$.
IR requires that all agents in $S'$ weakly prefer $a^*$ under $u+\tau'$ to $b$ under $u+\tau$, with at least one agent strictly better off. Because $b$ has full coverage, any such strict improvement for some agent must be financed by reducing the utility of other agents at $a^*$. The key step is to show that this redistribution cannot be supported without violating IR.

Lemma \ref{lem:observable_participation} implies that any agent whose utility at $a^*$ is reduced below their base utility must be actively participating in the deviation, i.e., must belong to $S'$. Hence, any agents outside $S'$ cannot supply the utility needed to fund a strict improvement for members of $S'$. This forces the total utility at $a^*$ under $u+\tau'$ to exceed the total utility of $b$ for agents in $S'$, contradicting full coverage, which guarantees that $b$ weakly dominates $a^*$ in total utility.
Thus, no IR-feasible transition from $(P,\tau,S)$ back to an $a^*$-implementing state exists. 
The full proof is presented in Appendix B.1.
\end{proofsketch}

As discussed above, a state $(P,\tau,S)$ can be destabilized in two ways. An \emph{outcome-changing} deviation alters reports so that $f(P') \neq f(P)$. An \emph{outcome-preserving} deviation leaves the selected alternative unchanged but allows a coalition to reallocate transfers so as to increase its members’ welfare while maintaining IR and full coverage. Lemma~\ref{lem:sw_max_swap} characterizes exactly when such outcome-preserving deviations are feasible.

The key insight is that any profitable redistribution must reduce the utility of some \emph{receiver} agent $i$ with $\tau_i(b) > 0$, where $b = f(P)$. Whether this reduction is feasible depends on the slack in $U_i(b)$ relative to $i$’s next-best alternatives. If $U_i(b)$ strictly exceeds all other utilities, slack is realized. Otherwise, slack is only potential and must be unlocked by adjusting transfers at tied next-best alternatives. By Lemma~\ref{lem:observable_participation}, any agent whose utility falls below their base utility must participate in the deviation. As a result, slack can be extracted from $i$ only if transfers at $i$’s next-best alternatives can be relaxed without forcing additional agents into the coalition. Lemma~\ref{lem:sw_max_swap} identifies three mutually exclusive scenarios in which this is possible, depending on whether slack is realized, supported by loose donors, or blocked by binding donors.

\begin{dfn}[Next-best alternative]
Given base utilities $u$, consider a state $(P, \tau, S)$ with $f(P) = b$ so that $b$ has full coverage under $u+\tau$.
For every $i \in \calN$, let 
\[
NBA_i = \arg\max_{a \neq b} U_i(a)
\] 
denote the set of next-best alternatives to $b$. 
\end{dfn}

\begin{dfn}[Binding NBA alternatives]
Given base utilities $u$, consider a state $(P, \tau, S)$ with $f(P) = b$ so that $b$ has full coverage under $u+\tau$.
Fix a receiver $i$, with $\tau_i(b) > 0$, and $\tau_i(x) > 0$ for all $x \in NBA_i$ such that $U_i(b) = U_i(x)$.
For any $x \in NBA_i$, let
\[
J_x = \{ j \in \calN : \tau_j(x) < 0 \};
\]
\[
\bar J_x = \{ j \in J_x : U_j(x) = U_j(b) \}.
\]
We denote by
\[
\calB_i = \{ x \in NBA_i : \bar J_x = J_x \}
\]
the set of \emph{binding} NBA alternatives where all donors are tight.
\label{dfn:binding_NBAs}
\end{dfn}

\begin{lem}
Given base utilities $u$ and consensus rule $f$, consider a state $(P,\tau,S)$ that is IR-feasible from the truthful state $(P^*,\emptyset,\emptyset)$.
Suppose that $f(P^*) = a^* \notin \arg\max_{a \in \calA} SW(a)$, that $f(P)=b \in \arg\max_{a \in \calA} SW(a)$, and that $b$ has full coverage under $u+\tau$. 
There exists an IR-feasible coalitional deviation $(P', \tau', S')$ with $f(P') \in \argmax_{a \in \calA} SW(a)$ if and only if $\exists i \in \calN$ with $\tau_i(b) > 0$ so that at least one of the following hold:
\begin{enumerate}
    \item[P1.] \emph{(Strict slack at $b$.)} \\ $U_i(b) > \max_{a \in \mathcal{A} \setminus \{b\}} U_i(a)$;
    \item[P2.] \emph{(Loose donors at every NBA.)}\\ $\forall x \in NBA_i$, both $\tau_i(x) > 0$ and $\exists j \in \calN$ so that $\tau_j(x)<0$ and $U_j(x) < U_j(b)$;
    \item[P3.] \emph{(Potential slack blocked only by binding donors.)} \\ Both:
    \begin{itemize}[leftmargin=0em]
        \item $\forall x \in NBA_i$, $\tau_i(x) > 0$;
        \item $\exists j \in \calN$ so that $\forall x \in \calB_i$, $\tau_j(x) < 0$ and $U_j(x) = U_j(b)$.
    \end{itemize}
\end{enumerate}
\label{lem:sw_max_swap}
\end{lem}

\begin{proofsketch}
The proof proceeds in two directions.

\emph{Sufficiency.}  
The three cases correspond to progressively weaker ways of extracting transferable slack from a receiver at $b$, while preserving IR and full coverage constraints 
throughout the deviation.
Fix a state $(P,\tau,S)$ with $f(P)=b$ and full coverage. Suppose $\exists i$ with $\tau_i(b)>0$ satisfying one of (P1)-(P3). In each case, we construct a small redistribution of transfers that lowers $U_i(b)$ by some $\epsilon>0$ and reallocates this slack to other agents, while preserving both IR for the deviating coalition and full coverage of $b$.

Under (P1), agent $i$ has \emph{realized} slack: $U_i(b)$ strictly exceeds all other alternatives. A donor at $b$ can therefore withdraw transfers from $i$ without changing $i$’s ranking, and capture the slack while maintaining IR.

Under (P2), slack is only \emph{potential}: $U_i(b)$ ties with each next-best alternative. However, at every such alternative $x \in NBA_i$ there exists a donor $j_x$, $\tau_{j_x}(x)<0$, whose slack is not tight, with $U_{j_x}(x) < U_{j_x}(b)$. These donors can each withdraw a small amount of utility from $i$ at their respective alternatives without violating IR or full coverage of $b$, collectively freeing slack at $b$ that can be redistributed.

Under (P3), some NBAs are \emph{binding}: all donors at those alternatives are tight. Slack can still be extracted if there exists a single agent who is tight across all binding alternatives. This agent can absorb the necessary adjustment across those alternatives and receive the released slack at $b$.

In all cases, the deviating coalition excludes $i$, so the reduction in $i$’s utility does not violate IR, and the modified transfers preserve full coverage of $b$. Hence an IR-feasible deviation with $f(P')=b$ exists.

\emph{Necessity.}  
Conversely, suppose there exists an IR-feasible deviation $(P',\tau',S')$ with $f(P')\in \argmax_{a \in \calA} SW(a)$. Since $b$ has full coverage, any strict gain for agents in $S'$ must be financed by reducing the utility of some non-participating receiver $k$ with $\tau_k(b) > 0$. By Lemma \ref{lem:observable_participation}, this reduction must preserve full coverage for $k$ at all alternatives.

If $k$ has realized slack at $b$, then (P1) holds. Otherwise, slack must come from $k$’s next-best alternatives. If at every such alternative there is a loose donor, then (P2) holds. If some alternatives are binding, feasibility requires a single agent who can absorb the required adjustments across all binding alternatives, yielding (P3). If none of these conditions holds, no IR-preserving redistribution can reduce $U_k(b)$, contradicting existence of the deviation.

Thus, conditions (P1)-(P3) are exhaustive and characterize exactly when an IR-feasible deviation preserving a welfare-maximizing outcome exists. The full proof is provided in Appendix~B.2.
\end{proofsketch}

Our results combine to yield a complete characterization of IR-SNE. The key conclusion is that the truthful state is an IR-SNE if and only if the truthful winner is a social-welfare maximizer. Intuitively, Lemma \ref{lem:no_Astar_transfer} rules out deviations that revert the outcome to a non-welfare-maximizing alternative once a welfare-maximizing outcome with full coverage is reached. Lemma \ref{lem:sw_max_swap} then characterizes when coalitions can profitably reallocate transfers while preserving the same welfare-maximizing outcome. Condition~3 of the theorem below encodes the slack constraints identified in Lemma \ref{lem:sw_max_swap}, ensuring that no such welfare-preserving deviation can exist.

Moreover, no IR-feasible deviation can select an outcome outside $\arg\max_{a \in \calA} SW(a)$. Any such deviation would necessarily violate full coverage. In particular, Lemma~5 in Appendix~B.1 implies that full coverage pins down utilities across all welfare-maximizing alternatives, ruling out IR-feasible deviations to non-welfare-maximizing outcomes.

\begin{thm}[Characterization of IR-Strong Nash Equilibria]
\label{thm:IR_SNE_characterization}
Fix base utilities $u$ and consensus rule $f$.
A state $(P,\tau,S)$ is an IR-SNE if and only if all of the following hold:
\begin{enumerate}
    \item \textbf{(IR-feasibility.)}
    $(P,\tau,S)$ is IR-feasible from the truthful state $(P^*,\emptyset,\emptyset)$
    (or $(P,\tau,S)=(P^*,\emptyset,\emptyset)$).
    \item \textbf{(Efficiency and coverage.)}
    $f(P) \in \arg\max_{a\in\mathcal{A}} SW(a)$ and $f(P)$ has full coverage under $u+\tau$.
    \item \textbf{(No profitable welfare-preserving deviation.)}
    For every $i \in \calN$ so that $\tau_i(f(P)) > 0$, none of the conditions of Lemma~\ref{lem:sw_max_swap} are satisfied.
\end{enumerate}
\end{thm}

\begin{algorithm}[t]
\caption{IR-SNE Construction for Consensus}
\label{alg:SW_stacked}
\begin{algorithmic}[1]
\REQUIRE Utilities $(u_i(a))_{i \in \mathcal{N},\, a \in \mathcal{A}}$, consensus rule $f$ with $f(P^*) = a^* \notin \argmax_{a \in \calA} SW(a)$
\ENSURE Transfer scheme $\tau$ with $(P, \tau, S)$ is an IR-SNE

\STATE Initialize $\tau_i(a) \leftarrow 0$ for all $i \in \mathcal{N}, a \in \mathcal{A}$
\STATE Order alternatives by social welfare: $1 \succeq 2 \succeq \cdots \succeq m$

\FORALL{$i \in \mathcal{N}$}
    \STATE $\tau_i(a^*) \leftarrow 0$; $U_i(a^*) \leftarrow u_i(a^*)$
\ENDFOR

\FORALL{$i \in \mathcal{N}$, $t = a^*+1, a^*+2 \ldots, m$}

    \STATE $U_i(t) \leftarrow \frac{SW(t)}{SW(a^*)} u_i(a^*)$; $\tau_i(t) \leftarrow U_i(t) - u_i(t)$

\ENDFOR

\FOR{$t = a^*-1, a^*-2, \ldots, 1$}

    \STATE Let $R_t = \{ i \in \mathcal{N} : u_i(t) < U_i(t+1) \}$; $D_t = \{ i \in \mathcal{N} : u_i(t) > U_i(t+1) \}$

    \FORALL{$i \in R_t$}
        \STATE $\tau_i(t) \leftarrow U_i(t+1) - u_i(t)$
    \ENDFOR

    \STATE Distribute total deficit 
    $-\sum_{i \in R_t} \tau_i(t)$
    among agents in $D_t$ so that for all $i \in D_t$,
    $u_i(t) + \tau_i(t) \ge U_i(t+1)$

    \FORALL{$i \in \mathcal{N}$}
        \STATE $U_i(t) \leftarrow u_i(t) + \tau_i(t)$
    \ENDFOR
\ENDFOR

\end{algorithmic}
\end{algorithm}

Theorem \ref{thm:IR_SNE_characterization} yields three transparent and easy-to-check conditions for verifying whether a state $(P,\tau,S)$ constitutes an IR-SNE. In particular, it separates the problem of equilibrium existence into feasibility, outcome selection, and the absence of profitable welfare-preserving deviations. This structure naturally suggests a constructive approach.

Algorithm \ref{alg:SW_stacked} exploits this characterization to explicitly construct an IR-SNE by iteratively building transfers that secure full coverage for a welfare-maximizing alternative while eliminating all exploitable slack. The following proposition establishes correctness and efficiency of this construction.

\begin{prop}
Given base utilities $u$ and consensus rule $f$ with default $a^* \in \calA$, consider the truthful state $(P^*, \emptyset, \emptyset)$ with $f(P^*) = a^* \notin \argmax_{a \in \calA} SW(a)$. Let $D_t = \{i \in \calN : u_i(t) > U_i(t+1)\}$. Algorithm \ref{alg:SW_stacked} yields a contract scheme $\tau$ so that $(P, \tau, S)$ is an IR-SNE, where $f(P) \in \argmax_{a \in \calA} SW(a)$ and
\[
S = \bigcup_{t < a^*} D_t \cup \left\{ i \in \calN : \argmax_{a \in \calA} U_i(a) \neq \argmax_{a \in \calA} u_i(a) \right\}
\]
is the set of agents that actively participate in the algorithm via donations or strategic voting. It runs in $\mathcal{O}(nm)$ time.
\label{prop:alg_proof}
\end{prop}

\begin{proofsketch}
Algorithm~\ref{alg:SW_stacked} constructs the transfer scheme $\tau$ in a stepwise manner to ensure that the welfare-maximizing alternative $b$ becomes the equilibrium outcome, while the default alternative $a^*$ receives no transfers.  

Intuitively, lines 3--5 assign $\tau_i(a^*) = 0$ for all $i \in \calN$, so that $a^*$ cannot benefit from additional contributions. Lines 6--8 ensure that $a^*$ weakly dominates every alternative with lower social welfare in $u+\tau$, so they do not pose a risk of deviation and can be ignored for equilibrium considerations. Since $a^*$ is not a social welfare maximizer, lines 9--18 then iterate through alternatives $t \in \{a^*-1, \ldots, 1\}$, partitioning agents into \emph{receivers} $R_t$ (who have smaller utility on $t$ than on $t+1$) and \emph{donors} $D_t$ (who have strictly greater utility). Receivers are granted transfers to make up their deficit, while donors contribute to satisfy budget balance and maintain full coverage of $t$.

By construction, any agent $i$ with $\tau_i(b) > 0$ does not satisfy any of the slack conditions (P1)--(P3) from Lemma~\ref{lem:sw_max_swap}. Combined with Lemma~\ref{lem:no_Astar_transfer}, which rules out IR-feasible deviations to $a^*$, and Lemma~5, which rules out deviations to non-welfare-maximizing alternatives, we conclude that no IR-feasible deviation exists.  

A full verification of conditions (P1)--(P3) for all agents is provided in Appendix~B.1.
\end{proofsketch}

\subsection{Stability Criteria for General AMR Rules}

Extending the characterization of IR-SNE beyond the consensus rule is challenging.
For many AMR rules, such as Maximin and Ranked Pairs, even deciding whether a fixed coalition can manipulate an outcome is NP-hard \citep{xia2009complexity}. Related hardness results extend to several bribery variants when the number of alternatives is unbounded \citep{faliszewski2008nonuniform,faliszewski2009hard}.
Rather than attempting a full characterization under a fixed voting rule, we instead identify a necessary condition for coalitional stability that is independent of the particular AMR rule and can be checked efficiently.

Proposition~\ref{prop:non_welfare_max_is_not_SNE} implies that any IR-SNE outcome must select a SW-maximizing alternative.
This leaves a natural question: \emph{even among welfare maximizers, when can a coalition reallocate transfers to induce a profitable deviation?}
To address this question, we introduce a notion that captures how much negative transfer can be reassigned without violating individual rationality.

\begin{dfn}[Reallocatable Amount]
\label{dfn:RA}
Given a state $(P,\tau,S)$ and an alternative $\hat a \in \mathcal{A}$, the \emph{reallocatable amount} to agent $j \in \mathcal{N}$ from a set of donors $S \subseteq \mathcal{N}$ over $\hat a$ is
\begin{align*}
\calRA(j,S,\hat a)
&=\sum_{i \in S} \max\{0,-\tau_i(\hat a)\}
+ \max\{0,-\tau_j(\hat a)\}.
\end{align*}
\end{dfn}

Intuitively, the reallocatable amount measures how much utility can be redirected toward a target agent from agents who are indifferent between the current outcome and a candidate alternative, together with the agent’s own negative transfers.
The following proposition uses this notion to characterize when an IR-feasible deviation exists under \emph{some} AMR rule.
Crucially, the voting rule is allowed to be chosen adversarially, isolating a purely economic constraint on stability.

\begin{prop}[IR-Stability via Indifferent Donors]
\label{prop:stability_via_indifferent_donors}
Fix base utilities $u$ and consider a state $(P,\tau,S)$ that is IR-feasible from the truthful state $(P^*,\emptyset,\emptyset)$.
Suppose that $f(P^*) = a^* \notin \arg\max_{a \in \calA} SW(a)$, that $f(P)=b \in \arg\max_{a \in \calA} SW(a)$, and that $b$ has full coverage under $u+\tau$.

Then there exists an AMR rule $f$ and an IR-feasible state $(P',\tau',S')$ with $f(P') = \hat a \in \calA \setminus \{b\}$ 
\emph{if and only if} there exists $j \in \mathcal{N}$ such that
\[
\calRA(j, S_j(\hat a), \hat a) > U_j(b) - U_j(\hat a),
\]
where
\[
S_j(\hat a) = \{ i \in \mathcal{N} \setminus \{j\} : U_i(b) = U_i(\hat a), \tau_i(\hat{a})<0 \}.
\]
\end{prop}

\begin{proofsketch}
The key idea is that any IR-feasible deviation toward an alternative $\hat a$ requires contributions from agents who are indifferent between $\hat a$ and the current winner $b$. 
The \emph{reallocatable amount} $\calRA(j,S_j(\hat a),\hat a)$ measures the maximum utility that can be redirected to a beneficiary $j$ using these indifferent donors and $j$'s own negative transfers. 
If this amount exceeds the difference $U_j(b) - U_j(\hat a)$, a suitably chosen AMR rule can make $\hat a$ the winner, providing sufficiency. Conversely, if no agent satisfies this inequality, any attempt to fund a gain at $\hat a$ would require contributions from agents who strictly prefer $b$, violating IR; this establishes necessity. 
Thus, the reallocatable-amount condition captures precisely when a coalition can profitably induce an alternative outcome under some AMR rule. The full proof is in Appendix B.2.
\end{proofsketch}

Proposition~\ref{prop:stability_via_indifferent_donors} demonstrates that full coverage alone is insufficient for IR-SNE. 
In addition to ensuring $U_i(f(P)) \geq U_i(f(P^*))$ for all agents, stability also requires that no agent can be subsidized via reallocations from indifferent donors to overturn the outcome. 
As an immediate corollary, Lemma~\ref{lem:no_Astar_transfer} follows by setting $\hat a = a^*$, since the reallocatable-amount condition fails for all agents.

When the AMR rule $f$ is fixed, passing the reallocatable-amount test is no longer sufficient to guarantee an IR-feasible deviation, reflecting the classical computational hardness of manipulation and bribery problems. 
For example, swap bribery remains NP-complete for $k$-approval with $k \ge 3$ \citep{elkind2009swap}. 
Nevertheless, the reallocatable-amount inequality serves as a polynomial-time filter for candidate deviations, ruling out many alternatives that cannot be sustained under any AMR rule.

\section{Discussion}


In this paper, we introduced an endogenous bribery mechanism that allows agents to coordinate plans in the presence of coalitional strategic behavior, along with a solution concept that is computationally simple and guaranteed to exist under the consensus rule. At the same time, our results expose several limitations and open problems that must be addressed before such mechanisms can be deployed in practice.

First, multi-agent planning is inherently combinatorial, since each agent may need to take different actions, leading to an exponential number of plans. As a result, computing IR-SNE may be intractable in realistic settings. This challenge motivates future work on simplifying the representation of the problem, exploiting structure in preferences, and developing approximation algorithms that can scale to larger instances.

Second, the welfare consequences of introducing an exchange market warrant further investigation. Although the proposed mechanism explicitly accounts for coalitional strategic behavior, the resulting equilibria may reduce overall efficiency or fairness compared to outcomes where agents act independently or within small cliques of agents \citep{kearns2008graphical}. These concerns are particularly salient when agents have unequal pricing power and when fairness is a key design consideration. Promising directions for future research include extending the framework to online settings, where agents or alternatives arrive and depart over time, as well as exploring auction-based mechanisms for contract formation.

Finally, our analysis adopts Pareto dominance as the basis for individual rationality, implicitly assuming that agents are willing to participate in any coalition that weakly benefits them. This abstraction overlooks the possibility that agents may object to outcomes in which others receive disproportionately large gains, particularly when they have strategic leverage through credible non-participation. For example, consider two agents with default payoffs $(0, 0)$ and two feasible outcomes: $(A) = (1, 1000)$ and $(B) = (2, 100)$. Although both outcomes are Pareto optimal, agent 1 may refuse to participate unless agent 2 agrees to outcome $(B)$, illustrating how leverage and commitment power can shape equilibrium selection. Modeling such behavior is closely related to the problem of optimal commitment, which is computationally hard and remains an important direction for future work \citep{littman2001implicit,conitzer2006computing}.












\clearpage
\bibliographystyle{named}
\bibliography{mybib}

\clearpage
\appendix

\input{forarxiv_camera_ready_suppl}

\end{document}

%% file: macros.tex
\usepackage{times}
\usepackage{soul}
\usepackage{url}
\usepackage{amsmath}
\usepackage{amsthm}
\usepackage{booktabs}
\usepackage{algorithmic}
\usepackage[ruled,vlined]{algorithm2e}

\usepackage{amsfonts}
\usepackage[inline]{enumitem}
\usepackage{etoolbox}
\usepackage{makecell}
\usepackage{mathtools}
\usepackage{multirow}
\usepackage[table]{xcolor} 
\usepackage{tabularx}
\usepackage{tabulary}
\newcolumntype{K}[1]{>{\centering\arraybackslash}p{#1}}
\usepackage{thmtools}
\usepackage{thm-restate}
\usepackage{dsfont}
\usepackage[textsize=footnotesize,color=red!25,bordercolor=red]{todonotes}
\usepackage{wasysym}
\usepackage[capitalise,noabbrev]{cleveref}
\Crefname{remark}{Remark}{Remarks}
\Crefname{rmk}{Remark}{Remarks}
\Crefname{dfn}{Definition}{Definitions}
\Crefname{thm}{Theorem}{Theorems}
\Crefname{coro}{Corollary}{Corollaries}
\Crefname{lem}{Lemma}{Lemmas}
\Crefname{examplex}{Example}{Examples}
\Crefname{prop}{Proposition}{Propositions}

\usepackage{tikz}
\usetikzlibrary{shapes,fit,calc}

\colorlet{mygray}{gray!40}

\declaretheoremstyle[
    bodyfont=\normalfont
]{normalstyle}

\declaretheorem[style=normalstyle,name=Definition]{dfn}
\declaretheorem[name=Theorem]{thm}

\declaretheorem[name=Lemma]{lem}
\declaretheorem[name=Proposition]{prop}

\declaretheorem[style=normalstyle,name=Example]{examplex}

\newenvironment{proofsketch}{%
  \proof}{\endproof}

\usepackage{color}

\usepackage{natbib}
\newcommand{\citeay}[1]{\citeauthor{#1} [\citeyear{#1}]}

\usepackage{svg} 
\usepackage{MnSymbol}

\newcommand{\calA}{\mathcal{A}}
\newcommand{\calB}{\mathcal{B}}

\newcommand{\calN}{\mathcal{N}}

\newcommand{\calL}{\mathcal{L}}
\newcommand{\calS}{\mathcal{S}}

\newcommand{\calRA}{\mathcal{RA}}

\DeclareMathOperator*{\argmax}{arg\,max}



%% file: forarxiv_camera_ready_suppl.tex
\section*{Appendix}

\section{Variant Models of Strategic Behavior}
\label{apx:strategic_behavior_variants}

\subsection{Sticky Strategic Behavior}
\label{apx:sticky}

In the main text, we analyzed a model of strategic behavior in which all agents in a coalition $S$ could simultaneously revise their contracts and votes, while agents outside the coalition $\mathcal{N} \setminus S$ updated their votes optimally in response to the new contract scheme. That is, in any coalitional deviation $(P', \tau', S')$ from a state $(P, \tau, S)$, the contract scheme $\tau'$ is derived from utilities $u+\tau$, and agents in $S'$ evaluate IR-feasibility assuming that all other agents vote truthfully with respect to $u+\tau'$.

Here, we consider a variant we call \emph{sticky} strategic behavior, where nonparticipating agents are even more passive. Intuitively, sticky behavior is closer to the standard notion of coalitional deviations in cooperative game theory: only participating agents act, while all others keep both their votes and contracts fixed. Formally, given a state $(P, \tau, S)$, a coalition $S'$ revises their votes $P'_{S'}$ and contracts $\tau'_{S'}$, while the remaining agents retain their previous votes and contracts. Equivalently, the updated vote profile is
\[
P' = (P'_{S'}, P_{-S'}),
\]
so that nonparticipating agents do not respond to changes in their utilities induced by $\tau'$. This shows that, unlike in the main model where nonparticipants best-respond to new contracts, sticky behavior can preclude IR-SNE, illustrating the importance of modeling how nonparticipating agents react to coalition deviations.

\begin{prop}
IR-SNE may fail to exist under sticky strategic behavior, even for two agents, two alternatives, and the consensus rule $f$.
\end{prop}

\begin{proof}
Consider $n=2$ agents, $m=2$ alternatives, and $f$ as the consensus rule with alternative $2$ as the default. Let utilities be
\[
u = \begin{pmatrix}
3 & 0 \\
0 & 1
\end{pmatrix}.
\]
Contracts $\tau$ can be parameterized by $\alpha, \beta \ge 0$:
\[
\tau = \begin{pmatrix}
-\alpha & \alpha \\
\beta & -\beta
\end{pmatrix},
\quad
\text{so that } u+\tau = \begin{pmatrix}
3-\alpha & \alpha \\
\beta & 1-\beta
\end{pmatrix}.
\]
A state $(P, \tau, S)$ is IR-feasible from the truthful state $(P^*, \emptyset, \emptyset)$, with $P^* = (2 \succ 1, 1 \succ 2)$ and $f(P^*) = 2$, if and only if
\[
\alpha \ge 1-\beta.
\]

Fix such $\alpha$ and $\beta$ and consider $S = \{1,2\}$ with votes
\[
P = (1 \succ 2, 1 \succ 2),
\]
so that $f(P) = 1$. This is an IR improvement over the truthful state since $\alpha \geq 1-\beta$ implies $3 - \alpha \geq \beta$.

Now consider a singleton coalition $S' = \{1\}$ under sticky behavior. Let agent 1 revise their contract to
\[
\tau'_{S'} = (0,0),
\quad \text{so that } u + \tau' = \begin{pmatrix} 3 & 0 \\ \beta & 1-\beta \end{pmatrix}.
\]
By the definition of sticky behavior, agent 2 keeps their previous vote, so $P' = P$ and $f(P') = 1$, strictly preferred by agent 1 over $(P, \tau, S)$.

Next, agent 2 can similarly deviate with $S'' = \{2\}$, $P'' = (1 \succ 2, 2 \succ 1)$, and $\tau'' = \emptyset$, which is strictly preferred by agent 2 over $(P', \tau', S')$ and IR-feasible. Since each agent can unilaterally deviate in turn, no IR-SNE exists. This illustrates that even in this minimal setting, sticky strategic behavior can preclude the existence of stable outcomes.
\end{proof}


\subsection{Anonymous Recipient Model}

Throughout this paper, we assume that agents’ contracts $(c_{i \rightarrow j}(a))_{i \in \calS, j \in \calN, a \in \calA}$ are unilaterally offered, pairwise payments between agents, conditioned on the alternative $a$ selected by the social choice rule $f$. We further assume that payments can be netted along transfer cycles. For example, if
\begin{itemize}
    \item $c_{i \rightarrow j}(a) = 2$,
    \item $c_{j \rightarrow k}(a) = 3$, and
    \item $c_{k \rightarrow i}(a) = 4$,
\end{itemize}
then these transfers can be reduced to
\begin{itemize}
    \item $c_{i \rightarrow j}(a) = 0$,
    \item $c_{j \rightarrow k}(a) = 1$, and
    \item $c_{k \rightarrow i}(a) = 2$,
\end{itemize}
respectively. In this reduced form, the induced net transfers satisfy $\tau_i(a) = 2$, $\tau_j(a) = -1$, and $\tau_k(a) = -1$. This representation allows us to reason about agents’ utilities using only the net transfer scheme $\tau$, and is therefore convenient for much of the analysis (e.g., Lemma~\ref{lem:observable_participation}).

However, care is required when defining coalitions solely in terms of the reduced-form transfers $\tau$, rather than the underlying bilateral contracts $c$. While $\tau$ is sufficient to evaluate utilities, it does not record which agents’ contractual actions change when a deviation occurs.

Recall that a coalition $S \subseteq \calN$ that changes the outcome of $f$ from $a$ to $b$ is individually rational if 
\[
u_i(b) + \tau_i(b) \geq u_i(a), \quad \forall i \in \calS,
\]
with strict inequality for at least one agent. Importantly, agents whose net transfers $\tau_i$ are unchanged may still need to belong to $S$ if their underlying contracts $c$ are reallocated to implement the deviation. Excluding such agents would incorrectly treat their contractual commitments as exogenously adjustable. The distinction between $c$ and $\tau$ is therefore essential for defining feasible coalitional deviations.

If contract formation does not specify recipient identities, so that agents can only commit to paying or receiving amounts without targeting particular agents, then agents whose role is limited to redistributing previously committed resources are not identified as participants in $S$. This corresponds to a centralized-intermediary interpretation in which only net transfers matter, and coalition membership is inferred solely from changes in $\tau$. As we show in Proposition~\ref{prop:non_existence_for_anonymous}, under this anonymous-recipient interpretation, IR-SNE may fail to exist. Our definition of coalition membership thus reflects contractual control rather than payoff changes alone.

\begin{prop}
IR-SNE may not exist when recipients are anonymous, even with three agents, two alternatives, and $f$ is the consensus rule.
\label{prop:non_existence_for_anonymous}
\end{prop}

\begin{proof}
Consider $n=3$ agents and $m=2$ alternatives, where $f$ selects alternative $2$ by default. Let
\[
u = \begin{pmatrix}
1 & 1 & 3+\epsilon \\
1 & 2 & 1
\end{pmatrix}.
\]
We construct four states, each specified by a net transfer scheme $\tau$ and the resulting utilities $u+\tau$:
\begin{itemize}
\item \textbf{State 1:} \[
\tau = \begin{pmatrix}
0 & 1 & -1 \\
-1 & 0 & 1
\end{pmatrix} \implies u+\tau = \begin{pmatrix}
1 & 2 & 2+\epsilon \\
0 & 2 & 2
\end{pmatrix}
\]
\item \textbf{State 2:} \[
\tau^1 = \begin{pmatrix}
0 & 1 & -1 \\
-1 & 1 & 0
\end{pmatrix} \implies u+\tau^1 = \begin{pmatrix}
1 & 2 & 2+\epsilon \\
0 & 3 & 1
\end{pmatrix}
\]
\item \textbf{State 3:} \[
\tau^2 = \begin{pmatrix}
0 & 2+\epsilon & -2-\epsilon \\
-1 & 1 & 0
\end{pmatrix} \implies u+\tau^2 = \begin{pmatrix}
1 & 3+\epsilon & 1 \\
0 & 3 & 1
\end{pmatrix}
\]
\item \textbf{State 4:} \[
\tau^3 = \begin{pmatrix}
0 & 2+\epsilon & -2-\epsilon \\
-1 & 0 & 1
\end{pmatrix} \implies u+\tau^3 = \begin{pmatrix}
1 & 3+\epsilon & 1 \\
0 & 2 & 2
\end{pmatrix}
\]
\end{itemize}
Starting from the truthful state $(P^*, \emptyset, \emptyset)$ with $f(P^*)=2$, State 1 is an individually rational deviation by the coalition $S=\calN$: agent 1 weakly benefits, agent 3 strictly improves, and agent 2 switches their vote to $(1 \succ 2)$, thus changing the outcome of $f$ to $1$.

From State 1, State 2 is individually rationally reachable without requiring any agent to be included in the active coalition. The transition involves only a redistribution of previously committed transfers, and no agent’s net utility decreases. Under the assumption of anonymous recipients, agents whose role is limited to such redistribution are not required to participate in the deviating coalition, making this transition feasible. The outcome of $f$ changes to $2$ as agent $2$ switches their vote back to $(2 \succ 1)$.

The remaining transitions among States 2, 3, 4, and back to 1 follow by analogous individually rational deviations, yielding an improvement cycle over outcomes $2,1,2$,  and $1$, respectively.

Although many other individually rational transfer schemes exist, it suffices to exhibit a finite set of states that block one another via individually rational deviations. The four states above form such a cycle under anonymous recipients, implying that no IR-SNE exists.
\end{proof}

\newpage

\section{Complete Proofs}

\subsection{Full Proof of Lemmas}
\label{apx:proof_of_lemmas}

\setcounter{lem}{1}
\begin{lem}
Given base utilities $u$ and consensus rule $f$, consider a state $(P,\tau,S)$ that is IR-feasible from the truthful state $(P^*,\emptyset,\emptyset)$.
Suppose that $f(P^*) = a^* \notin \arg\max_{a \in \calA} SW(a)$, that $f(P)=b \in \arg\max_{a \in \calA} SW(a)$, and that $b$ has full coverage under $u+\tau$. Then there is no state $(P', \tau', S')$ with $f(P') = a^*$ that is IR-feasible from $(P, \tau, S)$.
\end{lem}

\begin{proof}
Suppose for contradiction that there exists $(P', \tau', S')$ with $f(P') = a^*$ that is IR-feasible from $(P, \tau, S)$. We make use of the following three facts to complete this proof. 

First, by definition of IR,
\begin{equation}
U'_i(a^*) \geq U_i(b), \quad \forall i \in S',
\label{eq:no_Astar_1}
\end{equation}
where one agent inequality is strictly greater. 

Second, by definitions of IR and full coverage, and following Lemma \ref{lem:observable_participation},
\begin{equation*}
U_i(b) \geq u_i(a^*), \quad \forall i \in \calN,
\end{equation*}
Specifically, every $i \in S$ must have $U_i(b) \geq u_i(a^*)$, by definition of IR. Similarly, for every $i \in \calN \backslash S$, $U_i(b) \geq U_i(a^*) \geq u_i(a^*)$ since $\tau_i(a^*) \geq 0$, by Lemma \ref{lem:observable_participation}. This implies,
\begin{equation}
    \tau'_i(a^*) \geq 0, \quad \forall i \in S'.
    \label{eq:no_Astar_3b}
\end{equation}

Third, by definition of full coverage, 
\begin{equation}
    U_i(b) \geq U_i(a^*), \quad \forall i \in S'.
    \label{eq:no_Astar_4b}
\end{equation}
Hence, combining Equations (\ref{eq:no_Astar_1}) and (\ref{eq:no_Astar_4b}), we can denote by
\[
T = \sum_{i \in S'} \left( U'_i(a^*) - U_i(a^*) \right) > 0
\]
the utility agents in $S'$ gain for $a^*$ due to $\tau'$ above $\tau$. It follows that
\[
T > \sum_{\substack{i \in S' \\ \tau_i(a^*) < 0}} | \tau_i(a^*) |
\]
since even if every $i \in S'$ with $\tau_i(a^*) < 0$ fully recouped their donation by setting $\tau'_i(a^*) = 0$, this would only recover $u_i(a^*) \leq U_i(b)$ for each such agent, which is insufficient to satisfy the strict inequality in Equation (\ref{eq:no_Astar_1}). 

Let
\[
R = \sum_{\substack{i \in \calN \backslash S' \\ \tau_i(a^*)>0}} \tau_i(a^*)
\]
denote the gross amount of utility that agents in $\calN \backslash S'$ received for $a^*$ due to $\tau$. By Lemma \ref{lem:observable_participation}, agents in $\calN \backslash S'$ can contribute at most $R$ toward $T$ without violating the definition of $S'$: agents with $\tau_k(a^*) < 0$ cannot decrease $\tau'_k(a^*)$ further (case 1 of Lemma \ref{lem:observable_participation}), while agents with $\tau_k(a^*) \geq 0$ can reduce their transfers to at most $0$ before violating case 2.

There are now two cases. If $T > R$, then the required funding strictly exceeds what agents in $\calN \backslash S'$ can contribute without violating Lemma \ref{lem:observable_participation}, forming a contradiction.

If $T \leq R$, then agents in $\calN \backslash S'$ with $\tau_k(a^*) > 0$ 
have sufficient existing transfers to fund $T$ in principle. However, for 
any such funding to reach agents in $S'$, agents in $\calN \setminus S'$ 
must modify their outgoing contracts. By Definition 
\ref{dfn:coalitional_manipulation}, non-participating agents cannot modify 
their outgoing contracts, so any such redirection would force $k \in S'$, 
a contradiction.

In both cases, we obtain a contradiction. Therefore, there is no state $(P', \tau', S')$ with $f(P') = a^*$ that is IR-feasible from $(P, \tau, S)$.
\end{proof}

\setcounter{lem}{2}
\begin{lem}
Given base utilities $u$ and consensus rule $f$, consider a state $(P,\tau,S)$ that is IR-feasible from the truthful state $(P^*,\emptyset,\emptyset)$.
Suppose that $f(P^*) = a^* \notin \arg\max_{a \in \calA} SW(a)$, that $f(P)=b \in \arg\max_{a \in \calA} SW(a)$, and that $b$ has full coverage under $u+\tau$. 
There exists an IR-feasible coalitional deviation $(P', \tau', S')$ with $f(P') \in \argmax_{a \in \calA} SW(a)$ if and only if $\exists i \in \calN$ with $\tau_i(b) > 0$ so that at least one of the following hold:
\begin{enumerate}
    \item[P1.] \emph{(Strict slack at $b$.)} \\ $U_i(b) > \max_{a \in \mathcal{A} \setminus \{b\}} U_i(a)$;
    \item[P2.] \emph{(Loose donors at every NBA.)}\\ $\forall x \in NBA_i$, both $\tau_i(x) > 0$ and $\exists j \in \calN$ so that $\tau_j(x)<0$ and $U_j(x) < U_j(b)$;
    \item[P3.] \emph{(Slack can be extracted unless blocked by binding donors.)} \\ Both:
    \begin{itemize}
        \item $\forall x \in NBA_i$, $\tau_i(x) > 0$;
        \item $\exists j \in \calN$ so that $\forall x \in \calB_i$, $\tau_j(x) < 0$ and $U_j(x) = U_j(b)$.
    \end{itemize}
\end{enumerate}
\end{lem}

\begin{proof}
There are two directions to this proof: sufficiency and necessity. In Step 1, below, we prove the existence of an IR-feasible state $(P', \tau', S')$ with $f(P') \in \argmax_{a \in \calA} SW(a)$ for each of the three conditions of the lemma, (P1), (P2), and (P3). In Step 2, we subsequently prove that if none of (P1), (P2), or (P3) hold for any $i \in \calN$ with $\tau_i(b) > 0$, then $(P', \tau', S')$ cannot exist. 

\textbf{Step 1 ($\impliedby$):}
We first go through each of the three cases to show that if $\exists i \in \calN$ so that $\tau_i(b) > 0$ that satisfies either  (P1), (P2), or (P3), then there exists an IR-feasible coalitional deviation $(P', \tau', S')$ with $f(P') = b$. \\

\textbf{Case (P1).}
Suppose $\exists i \in \calN$ with $\tau_i(b) > 0$ so that
\[
U_i(b) = \max_{a \in \mathcal{A} \setminus \{b\}} U_i(a) + \beta
\]
for some $\beta > 0$. Since $\tau_i(b) > 0$ and $i$ experiences realized slack $\beta$, there exists at least one donor agent $j$ with $\tau_j(b) < 0$ who can shift part of their transfer to $i$ without violating IR. Consider the following modification of transfers $\tau'$ where:
\[
\tau'_k(a) = \begin{cases}
    \tau_i(b)-\frac{\beta}{2}, & k=i, a=b \\
    \tau_j(b)+\frac{\beta}{2}, & k=j, a=b \\
    \tau_k(a), & \text{otherwise}.
\end{cases}
\]
The active deviating coalition is $S' = \{j\}$. Only agent $j$ changes their transfer; all other agents, including $i$, remain bound by their existing contracts and vote truthfully according to the updated utilities. By construction, $b$ continues to have full coverage under $u+\tau'$, so the consensus rule selects $f(P') = b$.

Under the modified transfers $\tau'$, agent $j$'s utility strictly increases by $\frac{\beta}{2}$, while no agent in $S'$ experiences a loss. Although agent $i$'s utility decreases relative to $\tau$, $i \notin S'$ and does not participate in the deviation, so their loss does not affect the IR condition. Therefore, $(P', \tau', S')$ constitutes an IR-feasible deviation from $(P, \tau, S)$. \\

\textbf{Case (P2).}
Suppose $\exists i \in \calN$ with $\tau_i(b) > 0$ and $\tau_i(x)>0$, $\forall x \in NBA_i$. Suppose further that $\forall x \in NBA_i$, $\exists j_x \in \calN$ so that $\tau_{j_x}(x) < 0$ and $U_{j_x}(x) < U_{j_x}(b)$. Define $\epsilon > 0$ that is small enough so that:
\begin{itemize}
    \item $\epsilon <  U_i(b) - \max_{a \in \calA \setminus (NBA_i \cup \{b\})} U_i(a)$, the realized slack between $b$ and $NBA_i$, and the subsequent alternatives;
    \item $\epsilon < \min_{a \in NBA_i \cup \{b\}}\tau_i(a)$, the minimum amount $i$ receives over $b$ and each $NBA_i$;
    \item $\epsilon < \min_{x \in NBA_i}|\tau_{j_x}(x)|$, the minimum amount any $j_x$ donates to $i$ over $x \in NBA_i$.
\end{itemize}
Intuitively, these conditions on $\epsilon$ ensure that each donor $j_x$ continues to satisfy IR at $b$ after withdrawing $\epsilon$ of their offers from $i$ over each $x \in NBA_i$.
Let $j_b \in \calN$ be some agent donating to $i$ over $b$ so that $\tau_{j_b}(b) < 0$. Consider the following modification of transfers $\tau'$ where:
\[
\tau'_k(a) = \begin{cases}
    \tau_i(a)-\epsilon, & k=i, a\in NBA_i \cup \{b\} \\
    \tau_{j_b}(b)+\epsilon, & k={j_b}, a=b \\
    \tau_{j_a}(a)+\epsilon, & a \in NBA_i, k = j_a \\
    \tau_k(a), & \text{otherwise}.
\end{cases}
\] 
By construction, $b$ continues to have full coverage under $u+\tau'$, so the consensus rule selects $f(P') = b$. The active deviating coalition is $S' = \{j_b\} \cup {j_x}_{x \in NBA_i}$. The collection of $j_x$ agents withdraw their transfers away from $i$ for each alternative $x \in NBA_i$, while $j_b$ benefits by receiving $\epsilon$ more utility that $i$ surrendered. Although agent $i$'s utility decreases relative to $\tau$, $i \notin S'$ and does not participate in the deviation, so their loss does not affect the IR condition. All other agents, including $i$, remain bound by their existing contracts and vote truthfully according to the updated utilities.  Therefore, $(P', \tau', S')$ constitutes an IR-deviation from $(P, \tau, S)$. \\

\textbf{Case (P3).}
Suppose $\exists i \in \calN$ with $\tau_i(b) > 0$ and $\tau_i(x)>0$, $\forall x \in NBA_i$. Suppose further that $\exists \calB_i \subseteq NBA_i$ so that $\forall x \in \calB_i$, for every $j \in \mathcal{N}$ so that $\tau_j(x) < 0$, $U_j(x) = U_j(b)$. In particular, we assume that $\exists j^* \in \mathcal{N}$ so that $\forall x \in \calB_i$, $\tau_{j^*}(x) < 0$ and $U_{j^*}(x) = U_{j^*}(b)$. Thus, $j^*$ is the agent specified in condition (P3) of the lemma. This implies that $\forall x \in NBA_i \setminus \calB_i$, $\exists j_x \in \calN$ so that $\tau_{j_x}(x) < 0$ and $U_{j_x}(x) < U_{j_x}(b)$. 

Define $\epsilon > 0$ that is small enough so that:
\begin{itemize}
    \item $\epsilon <  U_i(b) - \max_{a \in \calA \setminus (NBA_i \cup \{b\})} U_i(a)$, the realized slack between $b$ and $NBA_i$, and the subsequent alternatives;
    \item $\epsilon < \min_{a \in NBA_i \cup \{b\}}\tau_i(a)$, the minimum amount $i$ receives over $b$ and each $NBA_i$;
    \item $\epsilon < \min_{x \in NBA_i \setminus \calB_i}|\tau_{j_x}(x)|$ and $\epsilon < \min_{x \in \calB_i}|\tau_{j^*}(x)|$, the minimum amounts any $j_x$ or $j^*$ donates to $i$ over $x \in NBA_i$.
\end{itemize}
Intuitively, these conditions on $\epsilon$ ensure that each donor $j_x$ continues to satisfy IR at $b$ after withdrawing $\epsilon$ of their offers from $i$ over each $x \in NBA_i$.
Consider the following modification of transfers $\tau'$ where:
\[
\tau'_k(a) = \begin{cases}
    \tau_i(a)-\epsilon, & k=i, a\in NBA_i \cup \{b\} \\
    \tau_{j_a}(a)+\epsilon, & a \in NBA_i \setminus \calB_i, k = j_a \\
    \tau_{j^*}(a) + \epsilon, & a \in \calB_i \cup \{b\}, k = j^* \\
    \tau_k(a), & \text{otherwise}.
\end{cases}
\] 
By construction, $b$ continues to have full coverage under $u+\tau'$, so the consensus rule selects $f(P') = b$. The active deviating coalition is $S' = \{j^*\} \cup \{j_x\}_{x \in NBA_i \setminus \calB_i}$. The collection of $j_x$ agents withdraw their transfers away from $i$ for each alternative $x \in NBA_i \setminus \calB_i$ Agent $j^*$ both withdraws their transfer from $i$ for each alternative in $x \in \calB_i$ and benefits by receiving $\epsilon$ more utility that $i$ surrendered. This differs from Case (P2) since $j^*$ is binding at some alternatives $x$, so they must be the unique agent to take $i$'s utility over $b$. Although agent $i$'s utility decreases relative to $\tau$, $i \notin S'$ and does not participate in the deviation, so their loss does not affect the IR condition. All other agents, including $i$, remain bound by their existing contracts and vote truthfully according to the updated utilities.  Therefore, $(P', \tau', S')$ constitutes an IR-deviation from $(P, \tau, S)$. \\

\textbf{Step 2 ($\implies$):}
We prove that the existence of an IR-feasible state $(P', \tau', S')$ with $f(P') \in \argmax_{a \in \calA} SW(a)$ implies at least one of (P1), (P2), (P3) by contraposition. In particular, we assume that $\forall i \in \calN$ with $\tau_i(b) > 0$ that none of (P1), (P2), or (P3) holds in order to demonstrate that such a state $(P', \tau', S')$ cannot exist.

As a preliminary case, suppose that $\tau_i(b)=0$ for all $i \in \calN$. 
Then for any alternative $c \in \arg\max_{a \in \calA} SW(a)$ we have $U_i(b)=U_i(c)$ for all $i$, by Lemma~\ref{lem:full_coverage_consensus}. 
For a deviation $(P',\tau',S')$ to satisfy $f(P')=c$, it must therefore hold that $U'_i(c)\ge U'_i(b)$ for all $i$, which implies $U'_i(c)=U'_i(b)$ for all $i$. 
Consequently, any nontrivial redistribution at $c$ would require some agent to donate at $b$, strictly lowering their utility at $b$ and violating IR feasibility. 
Hence no IR-feasible deviation selecting a SW-maximizer can exist in this case.

Henceforth, suppose in $(P, \tau, S)$ there exists at least one agent $i \in \calN$ so that $\tau_i(b) > 0$. We assume that $\forall i \in \calN$ with $\tau_i(b) > 0$, either one of the following cases holds:
\begin{itemize}
    \item \textbf{Case 1:} $\exists x \in NBA_i$ with $U_i(b) = U_i(x)$ and $\tau_i(x) \leq 0$;
    \item \textbf{Case 2:} Each of the following holds:
    \begin{itemize}
        \item $\forall x \in NBA_i$, $\tau_i(x) > 0$ and $U_i(x) = U_i(b)$;
        \item $\exists x \in NBA_i$, $\forall j \in \calN$ so that $\tau_j(x) < 0$, we have $U_j(x) = U_j(b)$; hence, $\calB_i \neq \emptyset$;
        \item $\bigcap_{x \in \calB_i} \bar{J}_x = \emptyset$.
    \end{itemize}
\end{itemize}
We now argue that Cases 1 and 2 exhaust all possibilities consistent with the negation of (P1), (P2), and (P3). 
Fix any agent $i \in \calN$ with $\tau_i(b)>0$. 
Since (P1) fails, $U_i(b)$ cannot be strictly greater than all other alternatives, and therefore there exists some $x \in NBA_i$ such that $U_i(b)=U_i(x)$. 
If $\tau_i(x)\le 0$ for some such $x$, then Case~1 holds. 
Otherwise, $\tau_i(x)>0$ for all $x \in NBA_i$. 
Failure of (P2) then implies that there exists at least one $x \in NBA_i$ such that every donor at $x$ is tight, i.e., $\bar J_x = J_x$, so that $\calB_i \neq \emptyset$. 
Finally, failure of (P3) implies that no single agent is a tight donor across all $x \in \calB_i$, or equivalently that $\bigcap_{x \in \calB_i} \bar J_x = \emptyset$, which is precisely Case~2.

Suppose for contradiction that $\exists (P', \tau', S')$ that is IR-feasible from $(P, \tau, S)$ with $f(P') = c \in \argmax_{a \in \calA} SW(a)$. By Lemma \ref{lem:full_coverage_consensus} we have $U_i(b) = U_i(c), \forall i \in \calN$, and $U'_i(b) = U'_i(c)$. Hence, without loss of generality, we can take $c=b$ since in order to update the distribution $(U'_i(c))_{i \in \calN}$ from $(U_i(c))_{i \in \calN}$, the distribution of $(U_i(b))_{i \in \calN}$ must be updated identically.

Then $U'_i(c) \geq U_i(b)$, $\forall i \in S'$, with at least one agent $h \in S'$ with $U'_h(c) > U_h(c)$. By budget balance, this implies $\exists k \in \calN \setminus S'$ so that $U'_k(b) < U_k(b)$ in order to support $h$'s utility growth. If $0 \geq \tau_k(b)$, then by Lemma \ref{lem:observable_participation} we would have $k \in S'$, a contradiction. Hence $k$ has $\tau_k(b) > 0$. There are now two possibilities: either $k$ satisfies Case 1 or Case 2. We consider these in sequence.

\textbf{Case 1.} Suppose $\exists x \in NBA_i$ with $U_i(b) = U_i(x)$ and $\tau_i(x) \leq 0$. 
Since $f(P') \in \argmax_{a \in \calA} SW(a)$, we must have $U'_i(f(P')) \geq U'_i(a)$, $\forall a \in \calA, \forall i \in \calN$. In particular, we require $U'_k(b) \geq U'_k(x)$, $\forall x \in NBA_i$. Since $\exists x \in NBA_k$ so that $U_k(b) = U_k(x)$ with $\tau_k(x) \leq 0$, this implies
\[
U_k(b) = U_k(x) > U'_k(b) \geq U'_k(x),
\]
so that $0 \geq \tau_k(x) > \tau'_k(x)$. Hence, $k \in S'$ by Lemma \ref{lem:observable_participation}, forming a contradiction.

\textbf{Case 2.}
Recall from Definition \ref{dfn:binding_NBAs} that
\[
\calB_k = \{x \in NBA_i : \forall j \in \calN ~(\tau_j(x) < 0 \implies U_j(x) = U_j(b))\}.
\]
Suppose that $\calB_k \neq \emptyset$ and that $\bigcap_{x \in \calB_i} \bar{J}_x = \emptyset$, where we recall
\[
\bigcap_{x \in \calB_i} \bar{J}_x =  \{j \in \calN : \tau_j(x) < 0 , U_j(x) = U_j(b) \}.
\]
This implies that for $x, y \in NBA_k$, where $x \neq y$, then $\forall j_x \in \bar{J}_x, \forall j_y \in \bar{J}_y$, all of the following holds:
\begin{itemize}
    \item $\tau_{j_x}(x) < 0$, $U_{j_x}(x) = U_{j_x}(b)$,
    \item $\tau_{j_y}(y) < 0$, $U_{j_y}(y) = U_{j_y}(b)$,
    \item either $\tau_{j_x}(y) \geq 0$ or $U_{j_x}(y) < U_{j_x}(b)$,
    \item either $\tau_{j_y}(x) \geq 0$ or $U_{j_y}(x) < U_{j_y}(b)$.
\end{itemize}
Hence, in order to maintain $U_k'(b) = U_k(b) - \epsilon$ for some $\epsilon > 0$, both $j_x$ and $j_y$ must claim $\epsilon$ in order to maintain $U'_{j_x}(b) \geq U'_{j_x}(a), \forall a \in \calN$, and $U'_{j_y}(b) \geq U'_{j_y}(a), \forall a \in \calN$. Clearly this cannot hold, so no such coalition $S'$ can contain both $j_x$ and $j_y$. Therefore $k$ cannot reduce their utility of $b$ by $\epsilon$ without violating $U'_k(b) \geq U'_k(a), \forall a \in \calN$, contradicting the existence of $(P', \tau', S')$.

This concludes the proof.
\end{proof}

\begin{lem}
The Lexicographic AMR Rule (Definition~\ref{dfn:lexicographic_AMR}) is anonymous, monotone, and resolute.
\label{lem:lex_rule_is_AMR}
\end{lem}

\begin{proof}
Let $f$ be the lexicographic rule with fixed ordering $O = (o_1, \dots, o_m)$, applied to any profile $P = (R_1, \dots, R_n) \in \calL(\calA)^n$.

\textbf{1. Resoluteness:}  
By construction, $f$ always selects the first alternative in $O$ that appears as the top-ranked choice of some agent.  
If all agents rank the same alternative $c$ at the top, $f(P) = c$.  
Otherwise, the scanning procedure through $O$ guarantees exactly one winner.  
Thus, $f$ is resolute.

\textbf{2. Anonymity:}  
The outcome depends only on which alternatives appear at the top of the rankings, not on which agents submitted which rankings.  
Permuting the agents (reordering $R_1, \dots, R_n$) does not change the set of top-ranked alternatives.  
Hence, $f$ is anonymous.

\textbf{3. Monotonicity:}  
Suppose $a \in \calA$ is the outcome under profile $P$, i.e., $f(P) = a$.  
If some agent raises $a$ in their ranking (possibly to the top), $a$ either remains the first top-ranked alternative in $O$ or appears earlier in the scanning order, so $f(P)$ cannot change away from $a$.  
Lowering any other alternative cannot hurt $a$, since the lexicographic scan prioritizes $a$ over alternatives lower in $O$.  
Thus, moving $a$ up in any agent’s ranking cannot decrease the chance that $f(P) = a$, satisfying monotonicity.

Since all three properties hold, the lexicographic rule is AMR.
\end{proof}

\begin{lem}
If an alternative $b$ has full coverage under $u+\tau$, then for every
$x \in \arg\max_{a \in \calA} SW(a)$ we have $U_i(b)=U_i(x)$ for all $i \in \calN$.
\label{lem:full_coverage_consensus}
\end{lem}

\begin{proof}
Full coverage implies $U_i(b)\ge U_i(x)$ for all $i$.
Summing over $i$ yields $SW(b)\ge SW(x)$.
Since $x$ is welfare-maximizing, equality must hold, and hence all individual
inequalities bind.
\end{proof}


\subsection{Full Proof of Propositions}
\label{apx:proof_of_props}

\setcounter{prop}{0}
\begin{prop}
Given base utilities $u$ and anonymous and monotonic rule $f$, consider the state $(P, \tau, S)$. If $f(P) \notin \argmax_{a \in \calA} SW(a)$, then $(P, \tau, S)$ is not an SNE.
\end{prop}

\begin{proof}
We prove the claim by constructing $(P', \tau', S')$ with $f(P') \in \argmax_{a \in \calA} SW(a)$, which is a feasible IR-deviation from $(P, \tau, S)$. Denote $c = f(P)$ and $b \in \argmax_{a \in \calA} SW(a)$. Define $\tau'$ so that:
\begin{itemize}
    \item $\tau'_i(c) = \tau(c),\quad \forall i \in \calN$;
    \item $U'_i(b) \geq U_i(c),\quad \forall i \in \calN$, with at least one agent having a strict inequality;
    \item $U'_i(b) \geq U'_i(a),\quad \forall i \in \calN, \forall a \in \calA \setminus\{b, c\}$. 
\end{itemize}
By monotonicity of $f$ and full coverage of $b$ under $u+\tau'$, we have $f(P') = b$. Since $SW(b) > SW(c)$, by assumption, such a transfer scheme $\tau'$ exists. Let $S'$ consist of all agents whose transfers are actively modified to achieve $U'_i(b)$ or strategically vote in favor of $b$. The state $(P', \tau', S')$ is IR with respect to $(P, \tau, S)$ since $U'_i(b) \geq U_i(c)$, $\forall i \in \calN$ with at least one agent having a strict inequality. Therefore $(P, \tau, S)$ is not an SNE.
\end{proof}

\setcounter{prop}{1}
\begin{prop}
Given base utilities $u$ and resolute consensus rule $f$, consider the truthful state $(P^*, \emptyset, \emptyset)$. If $f(P^*) \in \argmax_{a \in \mathcal{A}} SW(a)$, then $(P^*, \emptyset, \emptyset)$ is an SNE.
\end{prop}

\begin{proof}
Suppose for contradiction that there is an IR-feasible state $(P, \tau, S)$ so that $f(P) = c$, $f(P^*) = b$, and $b \neq c$. There are two ways for $c$ to become the winner when $f$ is consensus: either (I) $c$ is not the default outcome, but has full coverage under $u+\tau$, or (II) $c$ is the default outcome and $b$ has full coverage under $u$. We consider these in sequence.

\noindent \textbf{Case (I):}
Since $f$ is the consensus rule, we must have
\[
U_i(c) \geq U_i(b), \quad \forall i \in \calN.
\]
By definition of IR, we must have
\[
U_i(c) \geq u_i(b), \quad \forall i \in S
\]
where one agent has a strict inequality. Hence, $\exists \epsilon > 0$ so that some agent $i \in S$ has $U_i(c) \geq u_i(b)+\epsilon$.  Since transfers are budget balanced, $\sum_j U_j(c)=\sum_j u_j(c)\le\sum_j u_j(b)=\sum_j U_j(b)$. Combined with strict improvement for some $i\in S$, this implies that there exists $j$ such that $U_j(c)<U_j(b)$, since otherwise $\sum_j U_j(c) \ge \sum_j U_j(b)$. This contradicts full coverage of $c$ under $u+\tau$.

\noindent \textbf{Case (II):} By full coverage of $b$ under $u$, we have
\[
u_i(b) \geq u_i(a^*), \quad \forall i \in \calN.
\]
By definition of IR, we must have
\[
U_i(a^*) \geq u_i(b), \quad \forall i \in S
\]
where one agent has a strict inequality. This entails $\exists i \in S$ so that $U_i(a^*) > u_i(b) \geq u_i(a^*)$. Hence, $\tau_i(a^*) > 0$. Since all agents have zero initial transfers from the truthful state, $\exists j~:~\tau_j(a^*) < 0$. Hence, by lemma \ref{lem:observable_participation}, $j \in S$. Thus, $\exists j \in S$ so that $U_j(a^*) < u_j(b)$, forming a contradiction.
\end{proof}

\setcounter{prop}{2}
\begin{prop}
Given base utilities $u$ and consensus rule $f$ with default $a^* \in \calA$, consider the truthful state $(P^*, \emptyset, \emptyset)$ with $f(P^*) = a^* \notin \argmax_{a \in \calA} SW(a)$. Algorithm \ref{alg:SW_stacked} yields a contract scheme $\tau$ so that $(P, \tau, S)$ is an IR-SNE, where $f(P) \in \argmax_{a \in \calA} SW(a)$ and
\[
S = \bigcup_{t < a^*} D_t \cup \left\{ i \in \calN : \argmax_{a \in \calA} U_i(a) \neq \argmax_{a \in \calA} u_i(a) \right\}
\]
is the set of agents that actively participate in the algorithm via donations or strategic voting.
It runs in $\mathcal{O}(nm)$ time.
\end{prop}

\begin{proof}
Lines 3-5 of the algorithm ensure that $\tau_i(a^*) = 0$, $\forall i \in \calN$. Lines 6-8 ensure that all alternatives who have weakly less social welfare than $a^*$ are fully covered by $a^*$ in $u+\tau$. They can thus be overlooked for equilibrium considerations, without loss of generality. By assumption of the proposition, $a^*$ is not a SW maximizer, so lines 9-18 run for some alternatives $t \in \{a^*-1, \ldots, 1\}$. Line 10 partitions the set of agents who are \emph{receivers} of utility $R_t$ on alternative $t$, as those agents with strictly smaller utility on alternative $t$ than $t+1$, and \emph{donors} $D_t$, who have strictly greater utility. On lines 11-13, receivers $i \in R_t$ receive exactly $\tau_i(t) = U_i(t+1) - u_i(t)$ to make up for this lack, while donors make up for this deficit on line 146. Lines 15-17 account for budget balance for $\tau(t)$ across all agents. 

By Proposition \ref{prop:non_welfare_max_is_not_SNE}, we know that $(P^*, \emptyset, \emptyset)$ is not an SNE, which suggests that $(P, \tau, S)$ could conceivably be in equilibrium. Since $U_i(b) \geq U_i(a)$, $\forall a \in \calA$, and $U_i(b) \geq u_i(a^*)$, $\forall i \in \calN$, $(P, \tau, S)$ is IR-feasible from the truthful state $(P^*, \emptyset, \emptyset)$. We thus have to prove that there is no IR-feasible state $(P', \tau', S')$. For consistency of notation, denote $b =1 \in \argmax_{a \in \calA} SW(a)$. 

By Lemma \ref{lem:no_Astar_transfer}, there is no IR-feasible state $(P', \tau', S')$ with $f(P') = a^*$. There is clearly no IR-feasible state $(P', \tau', S')$ with $SW(f(P')) < SW(f(P))$ since $f$ is the consensus rule, following Lemma \ref{lem:full_coverage_consensus}. Hence, we must check that none of the three conditions of Lemma \ref{lem:sw_max_swap} hold for any $i \in \calN$ with $\tau_i(b) > 0$. We confirm this as follows.

\textbf{Check (P1).} Clearly alternative $b=1$ has full coverage in $u+\tau$. By line 12, we have that $U_i(b) = U_i(x)$ for some $x \in NBA_i$, $\forall i \in \calN$ so that $\tau_i(b) > 0$. Therefore (P1) does not hold for every such $i$.

\textbf{Check (P2) and (P3).} Notice that $\tau_i(a^*) = 0$, $\forall i \in \calN$, is assigned on lines 3-5 and that $U_i(t) = U_i(t+1)$ for all $t \in \{1, \ldots, a^*-1\}$ and $\forall i \in R_t$. Hence, $\forall i \in \calN$ with $\tau_i(b) > 0$, $\exists a_t \in \{2, \ldots, a^*\}$ so that $\tau_i(a) > 0$, $\forall a < a_t$, and $\tau_i(a_t) \leq 0$. Therefore, neither (P2) nor (P3) holds for any $i \in \calN$ with $\tau_i(b) > 0$.

Note that for each alternative, all agents are iterated through a linear number of times. Hence, the running time of the algorithm is $\mathcal{O}(nm)$.
\end{proof}



\setcounter{prop}{3}
\begin{prop}[IR-Stability via Indifferent Donors]
Fix base utilities $u$ and consider a state $(P,\tau,S)$ that is IR-feasible from the truthful state $(P^*,\emptyset,\emptyset)$.
Suppose that $f(P^*) = a^* \notin \arg\max_{a \in \calA} SW(a)$, that $f(P)=b \in \arg\max_{a \in \calA} SW(a)$, and that $b$ has full coverage under $u+\tau$.

Then there exists an AMR rule $f$ and an IR-feasible state $(P',\tau',S')$ with $f(P') = \hat a \in \calA \setminus \{b\}$ 
\emph{if and only if} there exists $j \in \mathcal{N}$ such that
\[
\calRA(j, S_j(\hat a), \hat a) > U_j(b) - U_j(\hat a),
\]
where
\[
S_j(\hat a) = \{ i \in \mathcal{N} \setminus \{j\} : U_i(b) = U_i(\hat a), \tau_i(\hat{a})<0 \}.
\]
\end{prop}

\begin{proof}
The proposition rests on the idea that indifferent donors, agents in $S_j(\hat{a})$, are those responsible for making $\hat{a}$ the winner from a state $(P, \tau, S)$ where $b \in \calA$ has full coverage in $u+\tau$. This follows from the definition of IR, where we require
\[
U'_i(\hat{a}) \geq U_i(b), \quad \forall i \in S',
\]
where at least one agent's inequality is strict. Intuitively, since $U_i(b) \geq U_i(\hat{a})$, $\forall i \in \calN$, this requires indifferent agents to participate. We claim that $\calRA(j, S_j(\hat{a}), \hat{a}) > U_j(b) - U_j(\hat{a})$ is precisely the condition needed for some agent $j$ to strictly prefer $\hat{a}$ in $\tau'$ to $b$ in $\tau$, thus providing sufficiency for an IR improvement. By making $j$ uncovered by $\hat{a}$ in $u+\tau'$, there is some AMR rule $f$ that returns $\hat{a}$ as the winner, as $j$ votes in favor of $\hat{a}$. We provide the formal proposition proof, as follows. \\

\noindent \textbf{Step 1 ($\impliedby$):} 
Let $f$ be the lexicographic rule, defined as follows. 

\begin{dfn}[Lexicographic Rule]
\label{dfn:lexicographic_AMR}
Fix a linear ordering $O = (o_1, o_2, \dots, o_m) \in \calL(\calA)$ of all alternatives.  
Let $P = (R_1, \dots, R_n) \in \calL(\calA)^n$ be a profile of rankings.  
Define $f : \calL(\calA)^n \to \calA$ as follows:
\begin{enumerate}
    \item If $top(R_j) = c$ for all $j \in \calN$, then $f(P) = c$.
    \item Otherwise, let $f(P)$ be the first alternative $o_k \in O$ such that $o_k$ appears as $top(R_j)$ for at least one $j \in \calN$.
\end{enumerate}
Ties are broken deterministically according to the fixed ordering $O$. Here, $top(R_j)$ denotes the top alternative in the ranking $R \in \calL(\calA)$, for some agent $j \in \calN$.
\end{dfn}

We prove the lexicographic rule is AMR in Lemma \ref{lem:lex_rule_is_AMR} in Appendix \ref{apx:proof_of_lemmas}. 
Suppose that $\exists j \in \calN, \hat{a} \in \calA \setminus \{b\}$ so that 
\begin{align*}
& \calRA(j, S_j(\hat{a}), \hat{a}) > U_j(b) - U_j(\hat{a}), \\
& S_j(\hat{a}) = \{i \in \calN \setminus \{j\} : U_i(b) = U_i(\hat{a}), \tau_i(\hat{a})<0\}.
\end{align*}
Let 
\[
\bar{S} = \{i \in \calN \setminus \{j\} : \tau_i(\hat{a})>0\}
\]
be a set of agents who are receivers over $\hat{a}$, whose payments summing to $\calRA(j, S_j(\hat{a}), \hat{a})$ can be reallocated. Define $(\beta_i)_{\beta \in \bar{S}}$ so that $\sum_{i \in \bar{S}} \beta_i = \calRA(j, S_j(\hat{a}), \hat{a})$, and $0 \leq \beta_i \leq \tau_i(\hat{a})$, $\forall i \in \bar{S}$.
Consider the following modifications of transfers $\tau'$ where:
\[
\tau'_k(a) = \begin{cases}
\tau_j(\hat{a}) + \calRA(j, S_j(\hat{a}), \hat{a}), & k=j, a = \hat{a}\\
\tau_i(\hat{a}) - \beta_i, & k \in \bar{S}, a = \hat{a}\\
\tau_k(a), & \text{otherwise}.
\end{cases}
\]
Note that if $\tau_j(\hat a)<0$, then $\calRA(j,S_j(\hat a),\hat a)$
includes $-\tau_j(\hat a)$, so $\tau'_j(\hat a)\ge 0$; if
$\tau_j(\hat a)\ge 0$, then $\calRA(j,S_j(\hat a),\hat a)$ does not
include any reclamation from $j$, and $\tau'_j(\hat a)$ simply increases
by the amount funded by reallocations from $\bar S$.

We define $S' = S_j(\hat{a}) \cup \{j\}$, where agents in $S_j(\hat{a})$ are neutral over $\hat{a}$ to this change and agent $j$ strictly benefits over $\hat{a}$. We have $P' = (R_i')_{i \in \calN}$ for $R_j' = \hat{a}$ and $R_i' = b, \forall i \neq j$. Hence, $f(P') = \hat{a}$ by Definition \ref{dfn:lexicographic_AMR}. Thus $(P', \tau', S')$ is an IR-feasible coalitional deviation $(P, \tau, S)$. \\

\noindent \textbf{Step 2 ($\implies$):} 
We prove the contraposition.
Assume that for all $i\in\mathcal N$ and all $a\in\calA$,
\begin{equation}
\calRA(i,S_i(a),a) \le U_i(b)-U_i(a).
\label{eq:RA_bound}
\end{equation}
Suppose, for contradiction, that there exists an AMR rule $f$ and an
IR-feasible coalitional deviation $(P',\tau',S')$ such that
$f(P')=\hat a \in \calA \setminus \{b\}$.
Since $f(P)=b$ and $f(P')=\hat a$, IR-feasibility implies the existence
of a strict beneficiary $j\in S'$ satisfying
\[
U'_j(\hat a) > U_j(b),
\]
or equivalently,
\begin{equation}
U'_j(\hat a)-U_j(\hat a) > U_j(b)-U_j(\hat a).
\label{eq:strict_gain}
\end{equation}
Let
\[
S_j(\hat a)
=
\{\, i\in\mathcal N\setminus\{j\} :
U_i(b)=U_i(\hat a),\ \tau_i(\hat a)<0 \,\}
\]
denote the set of agents who are indifferent between $b$ and $\hat a$
and who donate utility at $\hat a$.
By Definition~\ref{dfn:RA}, the quantity
$\calRA(j,S_j(\hat a),\hat a)$ is the maximum increase in $j$'s utility
at $\hat a$ that can be achieved using only:
(i) reallocation of transfers from agents in $S_j(\hat a)$, and
(ii) reclamation of $j$'s own negative transfers at $\hat a$.
Define the surplus
\[
\beta
=
\bigl(U'_j(\hat a)-U_j(\hat a)\bigr)
-
\calRA(j,S_j(\hat a),\hat a).
\]
By (\ref{eq:RA_bound}) and (\ref{eq:strict_gain}), we have $\beta>0$.
Hence, at least $\beta$ units of $j$'s gain at $\hat a$ must be funded by
means other than the reallocations accounted for by
$\calRA(j,S_j(\hat a),\hat a)$.

We now show that any such additional funding necessarily violates IR-feasibility. There are two exhaustive possibilities.

\medskip
\noindent\textbf{Option 1 (additional donation at $\hat a$).}
There exists an agent $k$ such that
\[
\tau'_k(\hat a) < \tau_k(\hat a) \le 0
\]
who  makes an additional donation at $\hat a$.
By Lemma~\ref{lem:observable_participation}, this implies $k\in S'$.
Moreover, by definition of full coverage of $b$ under $u+\tau$,
we have
\begin{equation}
U_k(b) > U_k(\hat a).
\label{eq:pref_b}
\end{equation}
The strict decrease in transfers at $\hat a$ yields
\begin{equation}
U'_k(\hat a) < U_k(\hat a).
\label{eq:utility_drop}
\end{equation}
Combining (\ref{eq:pref_b}) and (\ref{eq:utility_drop}) gives
$U'_k(\hat a) < U_k(b)$, violating IR-feasibility for agent $k$.

\medskip
\noindent\textbf{Option 2 (reallocation via non-indifferent donors).}
Otherwise, the surplus $\beta$ must be funded via reallocation of
existing transfers.
Thus, there exists an agent $\ell$ such that
\[
\tau_\ell(\hat a)>0
\quad\text{and}\quad
\tau'_\ell(\hat a)<\tau_\ell(\hat a),
\]
so that some of $\ell$'s incoming transfers at $\hat a$ are reduced.
This reduction must be effected by some existing donor agent $k \notin S_j(\hat{a})$ with $\tau_k(\hat{a}) < 0$ who necessarily belongs to $S'$.

Since $k\notin S_j(\hat a)$ and $b$ has full coverage under $u+\tau$, $U_k(b) > U_k(\hat a)$, which entails $U_k(b) > U'_k(\hat a)$. This contradicts IR-feasibility for agent $k$. \\

In both cases, we obtain a contradiction.
Therefore, no such IR-feasible deviation $(P',\tau',S')$ with
$f(P')=\hat a$ can exist, completing the proof.
\end{proof}